\documentclass[aps,pra,showpacs,nofootinbib]{revtex4}

\usepackage{amsmath,latexsym,amssymb,verbatim,enumerate,graphicx,amsfonts}
\usepackage{mathrsfs,hyperref,amsthm,braket,dsfont,graphicx,float}

\usepackage{bm}
\usepackage{tikz}
\usetikzlibrary{positioning,shapes}
\usepackage{color}

\theoremstyle{plain}
\newtheorem{theorem}{Theorem}[section]
\newtheorem{corollary}[theorem]{Corollary}

\newtheorem{conjecture}[theorem]{Conjecture}
\newtheorem{remark}[theorem]{Remark}

\newcommand{\tr}{{\rm{Tr}}}
\newcommand{\cN}{{\cal N}}

\begin{document}

\title{Quantum estimation through a bottleneck}

\author{Milajiguli Rexiti\footnote{Corresponding author. Electronic address: milajiguli.milajiguli@unicam.it}}
\affiliation{Department of Mathematics and Physics, Xinjiang Agricultural University, 
\"Ur\"umqi 830053,
China}
\affiliation{School of Advanced Studies, University of Camerino, 62032 Camerino, Italy}
\author{Stefano Mancini}
\email{stefano.mancini@unicam.it}
\affiliation{School of Science and Technology, University of Camerino, 62032 Camerino, Italy}
\affiliation{INFN-Sezione di Perugia, I-06123 Perugia, Italy} 

\begin{abstract}
{We study the estimation of a single parameter characterizing families of unitary transformations acting on two systems. We consider the situation with the presence of bottleneck, i.e. only one of the systems can be measured to gather information. The estimation capabilities are related to unitaries' generators. In particular, we establish continuity of quantum Fisher information with respect to generators. Furthermore, we find conditions on the generators to achieve the same maximum quantum Fisher information we would have in the absence of bottleneck. We also discuss the usefulness of initial entanglement across the two systems as well as across multiple estimation instances.}
\end{abstract}
\date{\today}
\maketitle

\section{Introduction}

{
Nowadays, it is clear that the ultimate detectability of signals and the accuracy with which their parameters can be estimated stem from quantum mechanical methods \cite{Hel}. Along this avenue, the estimation of a parameter characterizing states' unitary transformations has been widely studied  \cite{B05,H06,K07}. Ultimate limits for this case have been established by considering different strategies and by referring to the generators (optimal probe states were related to them) \cite{QM}.
In quantum mechanics unitaries are employed in ideal situations, however in practice, one has to deal with noisy states' transformation.
Then, more recently, quantum estimation has been lifted to quantum channel maps \cite{Sasaki, Fujiwara, Ji}.
The estimation of a quantum channel's parameter can be regarded as the estimation of the parameter characterizing the isometry behind it (representing its Stinespring dilation \cite{Stinespring}),
once only part of its image space is accessible (measurable). The unaccessible part is traced away and usually referred to as the environment. 
Moreover, if the dimension of the space where the isometry acts on (channel's input), is the same as the image space, the problem becomes of estimating a unitary, still with the restriction of partly 
unaccessible image space. 
Also in this context, it would be interesting to trace the estimation capabilities back to the unitaries' generators. 
This model resembles a bottleneck, a term often used in communication systems  to indicate a point in the enterprise where the flow of data is impaired since there is not enough data handling capacity to handle the current volume of data \cite{Boudec}.

Given a one-parameter family of unitaries $\{U_\alpha^{AE\to BF}\}$,  
we consider the estimation of the parameter $\alpha$ by accessing only the system $B$. 
This amounts to use the quantum channel ${\cal N}^{AE\to B}$ between $AE$ and $B$
of which $U_\alpha^{AE\to BF}$ represents the Stinespring dilation \cite{Stinespring}.
In Ref.\cite{RM19}, this scheme has been used to introduce the notion of ``privacy" in the quantum estimation framework. Here, we aim at relating the estimation capabilities, quantified by quantum Fisher information, to the unitaries' generators. 
The following issues will be addressed: Is quantum Fisher information continuous in terms of generators? Is it possible to achieve the same quantum estimation performance we would have in the absence of bottleneck? If not, what would be the gap?
}

{
We remark that besides the communication setting, it is a quite common situation where the
parameter to estimate is encoded into a larger
quantum state, while an experimentalist has access only to a
smaller subsystem (see e.g. \cite{Gambetta, Tsang, Alipour}). 
This is also inevitable in quantum field theory in curved space-time, because there are infinitely
many modes need to be traced over (see e.g. \cite{Dragan, Wang, Safranek}).
}

{
Since the quantum Fisher information will be the relevant tool, we first recall it in Section \ref{Sec:Fisher}, where we also detail the model to be studied.
Then, we establish continuity of quantum Fisher information with respect to the unitaries' generator in Section \ref{Sec:continuity}.
Conditions on the tensor product generators to achieve the same quantum Fisher information as in the absence of bottleneck are found in Section \ref{Sec:qubit1}. 
A recipe to analyze more complicated generators is illustrated in 
Section \ref{Sec:qubit2}.
As a main result, it turns out that accessing a restricted final system does not reduce estimation capabilities provided that the partial
trace of the generator over the accessed subsystem nullifies (a particular case with generator belonging
to the special unitary algebra is represented).
In these Sections, the usefulness of initial entanglement (across the two systems as well as across the multiple estimation instances) is also discussed.
Finally, in Section \ref{Sec:conclu} we draw our conclusions.}

{
\section{Basic Notions and Model}\label{Sec:Fisher}
}

In classical estimation theory, the optimal unbiased estimators of a parameter $\alpha$ are those saturating the Cramer-Rao inequality
\begin{equation}\label{eq:CR}
Var(\alpha) \geq \frac{1}{F(\alpha)},
\end{equation}
which establishes a lower bound on the mean square error (variance)
$Var(\alpha) =\mathbb{E}_{\alpha} (\hat\alpha-\alpha)^2=\mathbb{E}_{\alpha} (\hat\alpha
-\mathbb{E}_\alpha \hat\alpha)^2 $
of any unbiased estimator $\hat\alpha$. In other words, the Cramer-Rao inequality
establishes the ultimate bound on the precision of estimating the parameter $\alpha$.

In Eq.\eqref{eq:CR} $F(\alpha)$ is the Fisher Information 
defined as
\begin{equation}\label{eq:F}
F(\alpha):=\int p(\hat\alpha|\alpha) \left( \frac{\partial \ln p(\hat\alpha|\alpha)^2}{\partial \alpha}\right) d\hat\alpha=\int \frac{1}{p(\hat\alpha|\alpha)} \left( \frac{\partial p(\hat\alpha|\alpha)^2}{\partial \alpha}\right) d\hat\alpha, 
\end{equation}
where $p(\hat\alpha|\alpha)$ denotes the conditional probability of obtaining the value $\hat\alpha$ when the parameter has the value $\alpha$.

In quantum mechanics, we have  $p(\hat\alpha|\alpha)={\rm Tr}\left[\Pi(\hat\alpha)\rho(\alpha)\right]$, where ${\Pi(\hat\alpha)} $ are the elements of a positive operator-valued measure (POVM) and 
$\rho(\alpha)$ is the density operator parametrized by the quantity we want to estimate. Defining the Symmetric Logarithmic Derivative (SLD) $L_\alpha$ as the Hermitian operator satisfying 
\begin{equation}\label{eq:SLD}
\frac{L_\alpha \rho_\alpha+\rho_\alpha L_\alpha}{2}=\frac{\partial \rho_\alpha}{\partial \alpha},
\end{equation}
the Fisher Information \eqref{eq:F} can be rewritten as
\begin{equation} \label{eq:Fq}
F(\alpha)=\int \frac{{\rm Re}\left({\rm Tr}\left[\rho(\alpha)\Pi(\hat\alpha)L_\alpha\right]\right)^2}{{\rm Tr}\left[\rho(\alpha)\Pi(\hat\alpha)\right]}d\hat{\alpha}.
\end{equation}
{To evaluate the ultimate bounds to the precision of estimation, we should maximize \eqref{eq:Fq} over all} quantum measurements. However, we can easily get the following chain of inequalities
\begin{eqnarray}
F(\alpha) &\leq &\int \left| \frac{{\rm Tr}\left[\rho(\alpha)\Pi(\hat\alpha)L_\alpha\right]}{\sqrt{{\rm Tr}\left[\rho(\alpha)\Pi(\hat\alpha)\right]}}\right|^2 d\hat\alpha \label{intF1}\\
&=&\int \left|{\rm Tr} \left[
\frac{\sqrt{\rho(\alpha)}\sqrt{\Pi(\hat\alpha)}}
{\sqrt{{\rm Tr}\left[\rho(\alpha)\Pi(\hat\alpha)\right]}}
\sqrt{\Pi(\hat\alpha)}L_\alpha\sqrt{\rho(\alpha)}
\right]\right|^2 d\hat\alpha \label{intF2} \\
&\leq& \int {\rm Tr}\left[\Pi(\hat\alpha)L_\alpha \rho(\alpha)L_\alpha\right]d\hat\alpha
\label{intF3}\\
&=&{\rm Tr}\left[ \rho(\alpha)L_\alpha^2\right],\label{TrL}
\end{eqnarray}
{where the step from \eqref{intF2} to \eqref{intF3} is according to the Cauchy-Schwarz inequality applied with Hilbert-Schmidt scalar product of operators. Equations \eqref{intF1}-\eqref{TrL}
 show that the Fisher Information $F(\alpha)$ of any quantum measurement is upper} bounded by the so-called Quantum Fisher Information
\begin{equation}\label{QFI}
F(\alpha)\leq J(\alpha) :={\rm Tr}\left[ \rho(\alpha)L_\alpha^2\right]={\rm Tr}\left[\partial_\alpha \rho(\alpha)L_\alpha\right], 
\end{equation}
leading to the quantum Cramer-Rao bound \cite{Hel} 
\begin{equation}
Var(\alpha) \geq \frac{1}{J(\alpha)}.
\end{equation}
This holds true for single-shot measurement, while
for $N$ (independent) measurements the quantity $J(\alpha)$ on the r.h.s. must be multiplied by $N$.
The calculation of $J(\alpha)$ is doable because the SLD is given by a Lyapunov equation. However, it depends on the {probe states and hence should be maximized over them.} Hereafter we will indicate such a maximum by $\overline{J}$.

{A widely used model dealt with the estimation of a parameter $\alpha\in[0,2\pi]$} introduced into the system through a unitary transformation $U_\alpha=e^{-i\alpha G}$, being $G$ its generator.
In such a case the maximum Fisher information {over all probe states has been found as} \cite{QM}
\begin{equation}\label{qm}
\overline{J}=\left( \lambda_{max}-\lambda_{min}\right)^2,
\end{equation}
where $\lambda_{max},\lambda_{min}$ are the maximum and minimum eigenvalues of $G$. This is because the minimum error is achieved when the standard deviation of $G$ is maximum. In turn, this latter is achieved by preparing the probe in a state having maximum spread, i.e. equally-weighted superposition of the eigenvectors $\ket{\lambda_{max}}$ and $\ket{\lambda_{min}}$ of $G$ corresponding, respectively, to $\lambda_{max}$ and $\lambda_{min}$.\footnote{
In passing, we notice that 
the error defined in Eq.(1) of Ref.\cite{QM},
to be consistent with the results reported there,  
should have been written with a square root, i.e.
$\delta \varphi =\left\langle{\left(\varphi_{est}/| \frac{\partial \braket{\varphi_{est}}_{av}}{\partial \varphi}|-\varphi \right)^2}\right\rangle^{\frac{1}{2}}$.} 

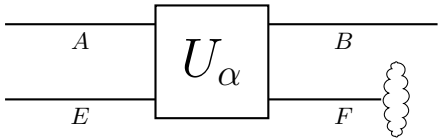
\begin{figure}[ht]
\begin{center}
\begin{tikzpicture}[scale=0.5]
\draw[thick] (0,0) -- (4,0); 
\draw[thick] (0,2) -- (4,2); 
\draw[thick] (7,2) -- (11.5,2);
\draw[thick] (7,0) -- (10,0);
\draw[thick] (4,-0.5) rectangle (7,2.5);
\node[cloud, cloud puffs=15.7, cloud ignores aspect, minimum width=0.3cm, 
minimum height=1cm, align=center, draw] (cloud) at (10.4,0){}; 
\node[left] at (0,0){$$};
\node[below] at (2,0){$E$};
\node[below] at (9,2){$B$};
\node[below] at (2,2){$A$}; 
\node[below] at (9,0){$F$};
\draw (5.5,1) node[font = \fontsize{20}{20}\sffamily\bfseries]{$U_\alpha$};
\end{tikzpicture}
\end{center}
\caption{Schematic representation of a unitary $U_\alpha:\mathscr{H}_A\otimes
\mathscr{H} _E\to\mathscr{H}_B\otimes\mathscr{H}_F$ whose parameter $\alpha$ has to be 
estimated by accessing only the system $B$.
{We refer to systems $A$ and $E$ (resp. $B$ and $F$) as the initial systems (resp. final systems). The systems $A$ and $E$ also constitute the input to the channel ${\cal N}^{AE\to B}$, while  $B$ is its output system.
}
 }\label{fig1}
\end{figure}

{ 
Suppose now to have the unitary $U_\alpha:\mathscr{H}_A\otimes
\mathscr{H} _E\to\mathscr{H}_B\otimes\mathscr{H}_F$ 
and consider the quantum channel
\begin{equation}\label{eq:calN}
\rho_{AE}\mapsto {\cal N}(\rho_{AE}) = 
{\rm Tr}_F \left[U_\alpha \rho_{AE} U_{\alpha}^\dag\right]
=\sum_\ell  {}_E\langle \ell | U_{\alpha} \rho_{AE} U_{\alpha}^\dag |\ell \rangle_E
=\sum_\ell K_\ell \rho_{AE} K_\ell^\dag,
\end{equation}
where the Kraus operators
\begin{equation}
K_\ell={}_E\langle\ell | U_{\alpha},
\end{equation}
depend on the parameter $\alpha$ 
(here $\{|\ell\rangle_E\}$ is an orthonormal basis of $\mathscr{H}_E$).

It is clear that the estimation of the parameter $\alpha$ characterizing the channel amounts to 
estimate the unitary $U_\alpha$ by accessing only the system $B$ (see Fig.\ref{fig1}).
This situation resembles a bottleneck. 
In a communications context, this happens when there isn't enough data handling capacity to handle the current volume of traffic. This is a common situation in network communication (for instance having a node with two incomes and one outcome links) \cite{Abbas}. After all, the depicted model describes a quantum multiple-access channel with two senders and one receiver \cite{QMAC}.

The idea of quantum estimation through a bottleneck gives rise to several questions, for example: Is quantum Fisher information continuous in terms of generators? Is it possible to achieve the same
quantum estimation performance we would have in the absence of bottleneck? If not what would be the gap?
Below we shall address these questions.
To simplify the treatment, we shall assume from now on 
$\mathscr{H}_A\simeq \mathscr{H}_B$ and $ \mathscr{H}_E\simeq \mathscr{H}_F$, 
as well as pure probe state $\rho$ on $AE$. 
}


\section{Continuity of quantum Fisher information}\label{Sec:continuity}

{
Quantum Fisher information has shown to 
be discontinuous in terms of the parameter to be estimated 
\cite{Dominik}.  
Discontinuities appear when, varying estimation parameter, 
 the rank of the density operator changes. 
The sudden drop is always connected to the information that
can be extracted from the change of purity and might also be 
a demonstration of a quantum phase
transition \cite{SM}.
}

{Here, on a different avenue,} we would like to address the issue of continuity of quantum Fisher information  related to {the state of} the system $B$ (see Fig.\ref{fig1}). 
This will make reliable numerical investigations of ${\overline J}_B$ whenever employed.\footnote{
Clearly, sampling a discontinuous function on a discrete set of points cannot be representative of the behavior of the function, while it can for a continuous function.
}
In particular, we would like to link this property to the generator of the unitary $U_\alpha$. {In the sense that two `close' generators (defined in some specific sense) should have, for the same probe state, `close'  quantum Fisher information.} 

The continuity property of quantum Fisher information has been established in Ref.\cite{Alireza},
{concerning both} the state $\rho(\alpha)$ and its derivative 
$\partial_\alpha\rho(\alpha)$. Here we {derive a slightly different version} of this result and then as a step further, we relate this issue to the generator of the unitary $U_\alpha$. 

\begin{theorem}\label{th:Jcontinuity}
Given two states $\rho_1(\alpha)$ and $\rho_2(\alpha)$ {on a finite dimensional Hilbert space 
$\cal H$ depending on a parameter $\alpha$, we have (dropping the dependence from 
$\alpha$ for a lighter notation):}
\begin{align}\label{eq:Jcontinuity}
\left|J(\rho_1)-J(\rho_2)\right| &\leq
\left[\frac{ \left\| \partial_\alpha \rho_1 \right\|{_2} \, \left\| \partial_\alpha \rho_2 \right\|{_2} }{2\lambda_1
(\lambda_1+\lambda_2)}
+\frac{ \left\| \partial_\alpha \rho_2 \right\|^2{_2} }{2\lambda_2
(\lambda_1+\lambda_2)}\right] \, \left\| \rho_1-\rho_2\right\|{_2}  \notag\\
&+
\left[\frac{ \left\| \partial_\alpha \rho_1 \right\|{_2} }{2\lambda_1}
+\frac{ \left\| \partial_\alpha \rho_2 \right\|{_2} }{
(\lambda_1+\lambda_2)}\right] \, \left\| \partial_\alpha\rho_1-\partial_\alpha\rho_2\right\|{_2} ,
\end{align}
where $\lambda_i\equiv\lambda_{min}(\rho_i)$. 
\end{theorem}

\begin{proof}
Let us start noticing that the solution of the Lyapunov Eq.\eqref{eq:SLD} can be written as 
\begin{equation}
L=2\int_0^\infty e^{-\rho t}\left(\partial_\alpha\rho\right) e^{-\rho t} dt,
\end{equation}
thus we have 
\begin{align}
\left|J(\rho_1)-J(\rho_2)\right| &= \left| \tr\left(L_1\partial{_\alpha}\rho_1\right)-
 \tr\left(L_2\partial{_\alpha}\rho_2\right) \right| \\
 &= \left| \tr\left\{ 2\int_0^\infty
\left[ \left(e^{-\rho_1 t}\partial{_\alpha}\rho_1\right)^2-
\left(e^{-\rho_2 t}\partial{_\alpha}\rho_2\right)^2 \right]dt\right\} \right| \\
 &\leq 2\int_0^\infty \left| \tr
\left[ \left(e^{-\rho_1 t}\partial{_\alpha}\rho_1\right)^2-
\left(e^{-\rho_2 t}\partial{_\alpha}\rho_2\right)^2 \right] \right| \, {dt} 
\label{Jdiff3}\\
&{= 2\int_0^\infty \left| \tr
\left[ \left(e^{-\rho_1 t}\partial{_\alpha}\rho_1-e^{-\rho_2 t}\partial{_\alpha}\rho_2\right)
\left(e^{-\rho_1 t}\partial{_\alpha}\rho_1+e^{-\rho_2 t}\partial{_\alpha}\rho_2\right)^2 
\right] \right|\, dt} \label{Jdiff4}\\
 &\leq 2\int_0^\infty \left\| 
e^{-\rho_1 t}\partial{_\alpha}\rho_1 - e^{-\rho_2 t}\partial{_\alpha}\rho_2 \right\|{_2}  \,
 \left\| e^{-\rho_1 t}\partial{_\alpha}\rho_1 + e^{-\rho_2 t}\partial{_\alpha}\rho_2 
 \right\|{_2} \, dt,
 \label{eq:integrand}
\end{align}
where {from \eqref{Jdiff3} to \eqref{Jdiff4}  we used the fact that the trace of a commutator vanishes, and from \eqref{Jdiff4} to \eqref{eq:integrand} we used
the Cauchy-Schwarz inequality applied with Hilbert-Schmidt scalar product of operators.} Let us now analyze separately the two terms entering in the integral 
\eqref{eq:integrand}. 
First, it is
\begin{align}
\left\| e^{-\rho_1 t}\partial{_\alpha}\rho_1 - e^{-\rho_2 t}\partial{_\alpha}\rho_2 \right\|{_2}
&=\left\| e^{-\rho_1 t}\partial{_\alpha}\rho_1 - e^{-\rho_1 t}\partial_\alpha\rho_2 
+ e^{-\rho_1 t}\partial_\alpha\rho_2-e^{-\rho_2 t}\partial{_\alpha}\rho_2 \right\|{_2}
\label{eq:eq0} \\
&\leq \left\| e^{-\rho_1 t} \right\|{_2} \,
\left\|\partial{_\alpha}\rho_1 - \partial_\alpha\rho_2\right\|{_2} 
+ \left\| e^{-\rho_1 t}-e^{-\rho_2 t} \right\|{_2}\, \left\|\partial{_\alpha}\rho_2 \right\|{_2} 
\label{eq:eq} \\
&\leq  \left\| e^{-\rho_1 t} \right\|{_2} \,
\left\|\partial{_\alpha}\rho_1 - \partial_\alpha\rho_2\right\|{_2} 
+ \left\|\rho_1-\rho_2\right\|{_2}\, t\,\int_0^1
\left\| e^{-\rho_1 t s} \right\|{_2} \, \left\|e^{-\rho_2 t (1-s)} \right\|{_2}\, ds\, 
\left\|\partial{_\alpha}\rho_2 \right\|{_2} \label{eq:diseq}\\
&{\leq  \left\| e^{-\rho_1 t} \right\|{_1} \,
\left\|\partial{_\alpha}\rho_1 - \partial_\alpha\rho_2\right\|{_2} 
+ \left\|\rho_1-\rho_2\right\|{_2}\, t\,\int_0^1
\left\| e^{-\rho_1 t s} \right\|{_1} \, \left\|e^{-\rho_2 t (1-s)} \right\|{_1}\, ds\, 
\left\|\partial{_\alpha}\rho_2 \right\|{_2} \label{eq:diseq1}}\\
&\leq  e^{-\lambda_1 t} 
\left\|\partial{_\alpha}\rho_1 - \partial_\alpha\rho_2\right\|{_2} 
+ \left\|\rho_1-\rho_2\right\|{_2}\, t\,\int_0^1
 e^{-\lambda_1 t s} \, e^{-\lambda_2 t (1-s)} \, ds\, 
\left\|\partial{_\alpha}\rho_2 \right\|{_2} \label{eq:diseq2} \\
&\leq  e^{-\lambda_1 t} 
\left\|\partial{_\alpha}\rho_1 - \partial_\alpha\rho_2\right\|{_2} 
+ \left\|\rho_1-\rho_2\right\|{_2}\, 
\frac{ e^{-\lambda_1 t} - e^{-\lambda_2 t}}{\lambda_2-\lambda_1} \, 
\left\|\partial{_\alpha}\rho_2 \right\|{_2},
\label{eq:1ineq}
\end{align}
{where, in going from \eqref{eq:eq0} to \eqref{eq:eq} we used the triangular inequality together with the sub-multiplicativity of Shatten's norms. 
From \eqref{eq:eq} to \eqref{eq:diseq} we used the property}
\begin{align}
\left\|e^{A}-e^{B}\right\|{_p} & = \left\| (A-B) \int_0^1 e^{(A-B)s} ds \,e^{B} \right\|{_p}  
\notag\\
& \leq \left\| (A-B) \right\|{_p}  \int_0^1 \left\|e^{As}\right\|{_p}  
\, \left\|e^{B(1-s)}\right\|{_p}  \, ds,  
\label{eq:propertyexp}
\end{align}
{valid for all $p$ such that $1\leq p \leq \infty$.
Next, the fact that $\|\cdot\|_1 \geq \|\cdot\|_2$ is employed from \eqref{eq:diseq} to \eqref{eq:diseq1}. Finally,  \eqref{eq:diseq2} and \eqref{eq:1ineq} immediately follow from the property of trace norm and by integration.}

For the other term in the integrand of Eq.\eqref{eq:integrand}, we have 
\begin{align}
\left\| e^{-\rho_1 t}\partial{_\alpha}\rho_1 + e^{-\rho_2 t}\partial{_\alpha}\rho_2 
\right\|{_2}
&\leq e^{-\lambda_1 t} \left\|\partial_\alpha\rho_1\right\|{_2}
+e^{-\lambda_2 t} \left\|\partial_\alpha\rho_2\right\|{_2}.
\label{eq:2ineq}
\end{align}
{At the end, plugging} \eqref{eq:1ineq} and \eqref{eq:2ineq} into \eqref{eq:integrand} we obtain
\begin{align}
\left|J(\rho_1)-J(\rho_2)\right|
\leq 2\int_0^\infty &\left[ 
e^{-\lambda_1 t} \left\|\partial{_\alpha}\rho_1 - \partial{_\alpha}\rho_2 \right\|{_2} 
+\left\|\rho_1-\rho_2\right\|{_2}
\frac{e^{-\lambda_1 t}-e^{-\lambda_2 t}}{\lambda_2-\lambda_1} \left\|\partial{_\alpha}\rho_2 \right\|{_2} \right] 
\notag\\
&\times\left[ 
e^{-\lambda_1 t} \left\|\partial{_\alpha}\rho_1\right\|{_2} 
+e^{-\lambda_2 t} \left\|\partial{_\alpha}\rho_2\right\|{_2} 
\right]dt,
\end{align}  
and after performing the integration we arrive at the desired result.
\end{proof}

\begin{corollary}\label{cor:continuityG}
{By referring to \eqref{eq:calN}, given $\rho_i(\alpha)=\tr_E\left[U_i(\alpha) \rho U_i^\dag(\alpha)\right] \equiv {\cal N}_i(\rho)$},
with $U_{i}(\alpha)=e^{-i\alpha G_i}$, 
we have
{(dropping the dependence from 
$\alpha$ for a lighter notation):}
\begin{align}
\left|J(\rho_1)-J(\rho_2)\right| &\leq
2\pi \left\{ C_1+C_2 \left({\rm dim}{\cal H}_E\right) 
\left[1+ 2\pi \left( \left\| G_1\right\|{_2} +  \left\| G_2\right\|{_2} \right) \right]\right\}
\left\| G_1 - G_2 \right\|{_2}
\end{align}
where, by referring to Theorem \ref{th:Jcontinuity}, we set 
\begin{align}
C_1 &:=
\left[\frac{ \left\| \partial_\alpha \rho_1 \right\|{_2} \, \left\| \partial_\alpha \rho_2 \right\|{_2}}{2\lambda_1
(\lambda_1+\lambda_2)}
+\frac{ \left\| \partial_\alpha \rho_2 \right\|^2{_2}}{2\lambda_2
(\lambda_1+\lambda_2)}\right], \\
C_2&:=
\left[\frac{ \left\| \partial_\alpha \rho_1 \right\|{_2}}{2\lambda_1}
+\frac{ \left\| \partial_\alpha \rho_2 \right\|{_2}}{
(\lambda_1+\lambda_2)}\right].
\end{align}
\end{corollary}

\begin{proof}
Concerning the first term at r.h.s. of Eq.\eqref{eq:Jcontinuity}, we have
\begin{align}
\left\|\rho_1-\rho_2\right\|{_2} &\leq  
\left\|\rho_1-\rho_2\right\|_1 \\
& \leq \left\|\cN_1- \cN_2\right\|_{\diamond}
\label{eq:diamondineq}\\
& \leq 2\inf_{V^F} \left\|\left(I^B\otimes V^F\right) U_1- U_2\right\|{_\infty}
\label{eq:contS}\\
& \leq 2\left\| U_1- U_2\right\|{_\infty}\\
& \leq 2 \alpha \left\| G_1- G_2\right\|{_\infty}\label{eq:UGineqprend}\\
& \leq 2 \alpha \left\| G_1- G_2\right\|{_2},
\label{eq:UGineq}
\end{align}
where \eqref{eq:diamondineq} follows from the fact that the diamond norm of a superoperator 
upper bounds its induced trace norm (see e.g. \cite{KMWY}), 
\eqref{eq:contS} comes from the continuity of the Stinespring dilation \cite{KSW},
and for \eqref{eq:UGineqprend} we have used the property \eqref{eq:propertyexp}. {Finally 
\eqref{eq:UGineq} results from $\|\cdot\|_p \leq \|\cdot\|_q$ for $p\geq q$.}

Regarding the second term at r.h.s. of Eq.\eqref{eq:Jcontinuity}, it is
\begin{align}
\left\|\partial_\alpha\rho_1- \partial_\alpha\rho_2\right\|{_2}
&=\left\| \tr_E\left(G_1 U_1\rho U_1^\dag- U_1\rho U_1^\dag G_1\right)
- \tr_E\left(G_2 U_2\rho U_2^\dag- U_2\rho U_2^\dag G_2\right) \right\|{_2} \\
&\leq  \, {\rm dim}{\cal H}_E \left\| G_1 U_1\rho U_1^\dag- U_1\rho U_1^\dag G_1
- G_2 U_2\rho U_2^\dag+ U_2\rho U_2^\dag G_2 \right\|{_2} 
\label{eq:rast}\\
&\leq  \, {\rm dim}{\cal H}_E \left\{ 
2\left\|G_1-G_2\right\|{_2} 
+\left(\left\| G_1\right\|{_2} + \left\| G_2\right\|{_2} \right) \, 
\left\| U_1\rho U_1^\dag- U_2\rho U_2^\dag\right\|{_2} \right\}
\label{eq:G1G2}\\
&\leq 2 \, {\rm dim}{\cal H}_E \left\{ 
\left\|G_1-G_2\right\|{_2} 
+ 2\left(\left\| G_1\right\|{_2} + \left\| G_2\right\|{_2} \right) \, \left\| U_1- U_2\right\|{_2} 
\right\}
\label{eq:UminusU}\\
&\leq 2 \, {\rm dim}{\cal H}_E \,  \left\|G_1-G_2\right\|{_2} 
\left\{ 1 + 2 \alpha \left(\left\| G_1\right\|{_2} + \left\| G_2\right\|{_2} \right) \right\}.
\label{eq:bounddiffderiv}
\end{align}
Eq.\eqref{eq:rast} is obtained by noticing that for any operator $O$ 
in ${\cal H}_A \otimes {\cal H}_E$,
we have $\|\tr_E O\|{_2} \leq  {\rm dim}{\cal H}_E \|O\|{_2}$ (see \cite{Rastegin}).
Eq.\eqref{eq:G1G2} follows by adding and subtracting terms $G_2U_1\rho U_1^\dag$ 
and $U_2\rho U_2^\dag G_1$ to the previous line. 

Finally, by inserting Eqs.\eqref{eq:UGineq}, \eqref{eq:bounddiffderiv} into \eqref{eq:Jcontinuity} and taking into account that $\alpha\in[0,2\pi]$ we get the desired result.
\end{proof}


\section{Two-qubit unitaries with tensor product generators}\label{Sec:qubit1}

Below we shall consider $\mathscr{H}_A\simeq\mathscr{H} _E\simeq\mathbb{C}^2$ with the aim of finding $\overline{J}_B$ and compare it to $\overline{J}_{BF}$ when $G$ can be written as $G_1\otimes G_2$.

Quite generally we can write 
\begin{align}
\label{eq:G1}
G_1&=\left(m_1 \sigma_1+m_2 \sigma_2+m_3 \sigma_3\right)+t_1 \sigma_0
\equiv\boldsymbol{\hat m  \cdot \sigma}+t_1 \sigma_0, \quad t_1\in\mathbb{R},\\
G_2&=\left (n_1 \sigma_1+n_2 \sigma_2+n_3 \sigma_3 \right)+t_2  \sigma_0
\equiv\boldsymbol{\hat n  \cdot \sigma}+t_2 \sigma_0, \quad t_2\in\mathbb{R},
\label{eq:G2}
\end{align}
where $\sigma_1,\sigma_2,\sigma_3$ are the Pauli operators and $\sigma_0=I$. Actually, {with no loss of generality} we can assume $m_1^2+m_2^2+m_3^2=n_1^2+n_2^2+n_3^2=1$.\footnote{We can always factor out e.g. 
$\|\hat{\boldsymbol{m}}\|$ from $G_1$, which will cause a rescaling of the parameter $t_1$, and incorporate it into the parameter $\alpha$. }

\begin{theorem}\label{th:D0}
Given a family of unitaries $U_\alpha=e^{-i\alpha G_1 \otimes G_2}$ with $G_1,G_2$ as in \eqref{eq:G1}, \eqref{eq:G2}, 
{we have} $\overline{J}_B=\overline{J}_{BF}$ iff $|t_1|\leq \min(1,|t_2|)$.
\end{theorem}

\begin{proof}
Let us introduce the eigenvectors of $\boldsymbol{\hat m  \cdot \sigma}$ as
\begin{align}
\boldsymbol{\hat m  \cdot \sigma}|0_{\,\boldsymbol{\hat m}}\rangle
&=|0_{\,\boldsymbol{\hat m}}\rangle.\\
\boldsymbol{\hat m  \cdot \sigma}
|1_{\,\boldsymbol{\hat m}}\rangle
&=-|1_{\,\boldsymbol{\hat m}}\rangle,
\end{align}
with
\begin{align}
|0_{\,\boldsymbol{\hat m}}\rangle&=\frac{1}{\sqrt{2}}
\left(\frac{m_1-im_2}{\sqrt{1-m_3}}\,
|0\rangle+\sqrt{1-m_3}\,|1\rangle
\right),\\
|1_{\,\boldsymbol{\hat m}}\rangle&=\frac{1}{\sqrt{2}}
\left(-\frac{m_1-im_2}{\sqrt{1+m_3}}\,
|0\rangle+\sqrt{1+m_3}\,|1\rangle
\right).
\end{align}
Here {$\{|0\rangle,|1\rangle\}$ is the basis of $\mathbb{C}^2$ 
consisting of the eigenvectors of $\sigma_3$.}
We can do similarly for $\boldsymbol{\hat n  \cdot \sigma}$.

The eigenvalues of $G_1\otimes G_2$ result  
$\{\left(t_1+1\right)\left(t_2+1\right), \left(t_1+1\right)\left(t_2-1\right), \left(t_1-1\right)\left(t_2+1\right), \left(t_1-1\right)\left(t_2-1\right)\}$, hence the maximum Fisher information we can get when accessing {the whole final system is, according to \eqref{qm}:}
\begin{equation}\label{eq:maxJBF}
\overline{J}_{BF}=\left\{ \begin{array}{ccc}
4\left(1+|t_1|\right)^2 & & |t_2|\leq|t_1|, \; |t_2|<1 \\
4\left(1+|t_2|\right)^2 & & |t_1|\leq|t_2|, \; |t_1|<1 \\
4\left(t_1+t_2\right)^2 & & |t_1|, |t_2|\geq 1, \; t_1t_2>0\\
4\left(t_1-t_2\right)^2 & & |t_1|, |t_2|\geq 1, \; t_1t_2<0\\
\end{array}\right..
\end{equation}
In the basis $\{ \ket {0_{\boldsymbol{\hat m}}} \ket {0_{\boldsymbol{\hat n}}}, 
\ket {0_{\boldsymbol{\hat m}}}\ket {1_{\boldsymbol{\hat n}}},
\ket {1_{\boldsymbol{\hat m}}}\ket {0_{\boldsymbol{\hat n}}},
\ket {1_{\boldsymbol{\hat m}}}\ket {1_{\boldsymbol{\hat n}}}
\}$ we can write the generic input state as
\begin{equation}\label{eq:input}
C_{00}|0_{\,\boldsymbol{\hat m}}\rangle
|0_{\,\boldsymbol{\hat n}}\rangle
+C_{01}|0_{\,\boldsymbol{\hat m}}\rangle
|1_{\,\boldsymbol{\hat n}}\rangle
+C_{10}|1_{\,\boldsymbol{\hat m}}\rangle
|0_{\,\boldsymbol{\hat n}}\rangle
+C_{11}|1_{\,\boldsymbol{\hat m}}\rangle
|1_{\,\boldsymbol{\hat n}}\rangle,
\end{equation}
with $C_{ij}\in \mathbb{C}$, such that $\sum_{i,j=0}^1 |C_{ij}|^2=1$. 
In turn, the unitary reads:
 \begin{equation}\label{eq:Umn}
U(\alpha)=e^{-i\alpha G_1\otimes G_2}=diag\{ e^{-i \alpha  (t_1+1) (t_2+1)}, e^{-i \alpha  (t_1+1) (t_2-1)}, e^{-i \alpha  (t_1-1) (t_2+1)}, e^{-i \alpha  (t_1-1) (t_2-1)}\}.
\end{equation}
Applying \eqref{eq:Umn} to \eqref{eq:input} and 
tracing away $F$ yields
\begin{equation}
\rho_B=\left(\begin{array}{cc}
|C_{00}|^2+|C_{01}|^2 &  e^{-2 i(1+t_2) \alpha }( C_{00}C_{10}^{*} +C_{01}C_{11}^{*}e^{4 i \alpha })\\ \\
e^{2 i(1+t_2) \alpha }( C_{00}^{*} C_{10}+C_{01}^{*}C_{11}e^{-4 i \alpha }) & |C_{10}|^2+|C_{11}|^2 
\end{array}\right).
\end{equation}
The Fisher information of $\rho_B$ can be evaluated,
using the methods of Sec.\ref{Sec:Fisher}, as
\begin{eqnarray}\label{eq:JBsimple}
J_B&=&16\frac{ |\left(|C_{01}|^2(t_2-1)-|C_{00}|^2(t_2+1)\right) C_{10} C_{11}^{*}e^{4 i a}-\left(|C_{11}|^2(t_2-1)-|C_{10}|^2(t_2+1)\right)C_{00} C_{01}^{*} |^2}{ |C_{00} C_{11}-C_{01} C_{10} e^{4 i a }|^2 }\\ \notag
&-&16\frac{\left({C_{00}^{*}}^{2}C_{01}^{2}C_{10}^{2} {C_{11}^{*}}^{2}e^{-8 i a}+C_{00}^{2}{C_{01}^{*}}^{2}{C_{10}^{*}}^{2}C_{11}^{2}e^{8i a}\right)t_2-2|C_{00}|^2|C_{01}|^2|C_{10}|^2|C_{11}|^2 t_2^2}{|C_{00} C_{11}-C_{01} C_{10} e^{4 i a }|^2 }.
\end{eqnarray}
In order to eventually attain the value of Eq.\eqref{eq:maxJBF}, $J_B$ should not depend on $\alpha$.
This implies to have $C_{00}=0 \vee C_{01}=0 \vee C_{10}=0 \vee C_{11}=0$. 
{As a consequence,} the maximum of $J_B$ will be 
\begin{equation}\label{eq:maxcond}
\overline{J}_B=\left\{ \begin{array}{ccc}
4\left(t_2+1 \right)^2 & \text{for}\quad  t_2\geq0 &\text{with}\quad |C_{00}|=|C_{10}|=\frac{1}{\sqrt{2}}  \\
4\left(t_2-1 \right)^2 & \text{for}\quad  t_2\leq0 &\text{with}\quad |C_{01}|=|C_{11}|=\frac{1}{\sqrt{2}}
\end{array}\right..
\end{equation}
The gap between $\overline{J}_B$ and $\overline{J}_{BF}$ then reads
\begin{equation}\label{eq:gap1}
\Delta:=\overline{J}_{BF}-\overline{J}_B=\left\{ \begin{array}{ccc}
0 & & |t_1| \leq |t_2|, \; |t_1| \leq 1 \\
4\left( |t_1|-|t_2|\right)\left( 2+|t_1|+|t_2|\right) & & |t_2| \leq |t_1|, \; |t_2| \leq 1 \\
4\left( t_1-1\right)\left( t_1+2|t_2|+1\right) & &  t_1>1, \; |t_2|>1 \\
4\left( t_1+1\right)\left( t_1+2|t_2|-1\right) & & t_1<-1, \; |t_2|>1
\end{array}\right..
\end{equation}
\hfill
\end{proof}
\begin{remark}
According to the conditions \eqref{eq:maxcond}, the maximum of $J_B$ is achieved by separable states.
In other words entangled input is not useful for this task.
\end{remark}

\begin{corollary}
Given a family of two-qubit unitaries $U^{AE\to BF}_\alpha=e^{-i\alpha G_1\otimes G_2}$,
with generator $G_i=\sum_{j=0}^3c^{(i)}_{j}\,\sigma_j$ 
($\sigma_0=I$ and $c^{(i)}_{j}\in\mathbb{R}$),
{to have} $\overline{J}_B=\overline{J}_{BF}$ it is sufficient that 
$c^{(1)}_{0}=0$, i.e.
${\rm Tr} G_1=0$ or equivalently $G_1\in\mathfrak{su}(2)$.
\end{corollary}

\begin{proof}
It immediately follows {from} Theorem \ref{th:D0} by observing that $G_1\in \mathfrak{su}(2)$ iff $t_1=0$.
\end{proof}


\subsection{Multiple instances estimation}\label{Sec:Multiple}

Here we shall consider estimation by multiple copies of the unitary $U_\alpha$. This will allow us to investigate the usefulness of entanglement across inputs on different copies of $U_\alpha$.
The simplest and non-trivial case is represented by two copies of the unitary $U_\alpha$ 
as depicted in Fig.\ref{model2}.

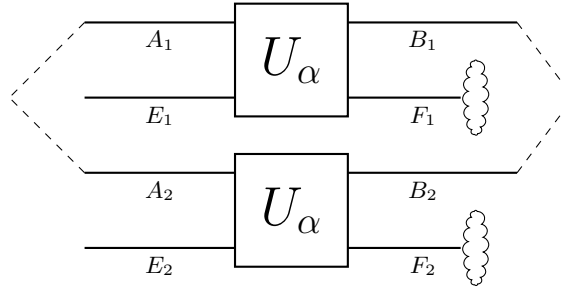
\begin{figure}[ht]
\begin{center}
\begin{tikzpicture}[scale=0.5]
\draw[thick] (0,0) -- (4,0); 
\draw[thick] (0,2) -- (4,2); 
\draw[thick] (7,2) -- (11.5,2);
\draw[thick] (7,0) -- (10,0);
\draw[thick] (4,-0.5) rectangle (7,2.5);
\node[cloud, cloud puffs=15.7, cloud ignores aspect, minimum width=0.3cm, minimum height=1cm, align=center, draw] (cloud) at (10.4,0){}; 
\node[left] at (0,0){$$};
\node[below] at (2,0){$E_1$};
\node[below] at (9,2){$B_1$};
\node[below] at (2,2){$A_1$}; 
\node[below] at (9,0){$F_1$};
\draw (5.5,1) node[font = \fontsize{20}{20}\sffamily\bfseries]{$U_\alpha$};
\draw[thick] (0,-2) -- (4,-2); 
\draw[thick] (0,-4) -- (4,-4); 
\draw[thick] (7,-2) -- (11.5,-2);
\draw[thick] (7,-4) -- (10,-4);
\draw[thick] (4,-4.5) rectangle (7,-1.5);
\node[cloud, cloud puffs=15.7, cloud ignores aspect, minimum width=0.3cm, minimum height=1cm, align=center, draw] (cloud) at (10.4,-4){}; 
\node[left] at (0,-2){$$};
\node[below] at (2,-4){$E_2$};
\node[below] at (9,-2){$B_2$};
\node[left] at (0,1){$$}; 
\node[right] at (10,1){$$}; 
\node[below] at (2,-2){$A_2$}; 
\node[below] at (9,-4){$F_2$};
\draw (5.5,-3) node[font = \fontsize{20}{20}\sffamily\bfseries]{$U_\alpha$};
\draw[dashed] (0,2) -- (-2,0);
\draw[dashed](0,-2)--(-2,0);
\draw[dashed] (11.5,2) -- (13,0);
\draw[dashed](11.5,-2)--(13,0);
\end{tikzpicture}
\end{center}
 \caption{Estimation of $\alpha$ by two copies of the unitary $U_\alpha$. The input system $A_1$ can be entangled with $A_2$.} \label{model2}
\end{figure}

Assuming $U_\alpha=e^{-i \alpha G_1\otimes G_2}$ with $G_1,G_2$ given by Eqs.\eqref{eq:G1}
and \eqref{eq:G2}, we know from conditions \eqref{eq:maxcond} that the optimal input states on single instance are 
\begin{equation}
\ket{\Psi(\phi)}_{i} = \frac{\ket{0_{\boldsymbol{\hat m}}}+e^{i \phi}\ket{1_{\boldsymbol{\hat m}}}}{\sqrt 2}
\ket{0_{\boldsymbol{\hat n}}}, \qquad \phi\in[0,2\pi],
\label{eq:optt2}
\end{equation}
for $t_2\geq 0$.
Then we expect the optimal input in two instances to be an entangled state built up with the twofold tensor product of states \eqref{eq:optt2}.
Let us consider 
\begin{align}\label{input2}
\ket \Upsilon_{i}=\frac{1}{\sqrt{2}} \left(\ket{\Psi(0)}_{i}^{\otimes 2} +
\ket{\Psi(\pi)}_{i}^{\otimes 2} \right),
\end{align}
that {is maximally entangled between $A_1$ and $A_2$.}

The {(global) final state} after unitary transformation reads
\begin{eqnarray}
\ket \Upsilon_{f}= \left( U_{\alpha}\otimes U_{\alpha}\right) \ket\Upsilon_{i},
\end{eqnarray}
where $U_\alpha$ is given by \eqref{eq:Umn}.
The maximum Fisher information we can get when accessing the {whole final system} is
4 times of the one in Eq.\eqref{eq:maxJBF}.
The output state we are going to measure is
\begin{eqnarray}
\rho_{\boldsymbol{B}}={\rm Tr}_{\boldsymbol{F}} |\Upsilon\rangle_{f}\langle \Upsilon |,
\end{eqnarray}
with $\boldsymbol{B}:=\left( B_1, B_2\right), \boldsymbol{F}:=\left( F_1, F_2\right)$.
The Fisher information with this state, 
computed according to the methods of Sec.\ref{Sec:Fisher}, results 
\begin{equation}
{J}_{\boldsymbol{B}}=16 (1+t_2)^2.
\end{equation}
Repeating the {above steps} for $t_2<0$,
which amounts to flip the environment state
$\ket{0_{\boldsymbol{\hat n}}}$ to $\ket{1_{\boldsymbol{\hat n}}}$,
{in} \eqref{eq:optt2} and hence {in \eqref{input2},}
we can conclude that the gap between $\overline{J}_{\boldsymbol{BF}}$ and 
$\overline{J}_{\boldsymbol{B}}$ reads as 4 times (number of instances squared) of the one in \eqref{eq:gap1}.
This amounts to have $\Delta=0$ for two instances estimation as well, under conditions \eqref{eq:gap1}.


\section{Two-qubit unitaries with generic generators}\label{Sec:qubit2}

{
The aim of this Section is to provide a procedure to find $\overline{J}_B$ for any given generator.
We shall then apply the procedure to a case study that, though is not the most general, it is enough representative to draw some general conclusions.
}

According to Ref.\cite{QM} if we consider a single qubit unitary transformation $U=e^{-i\alpha G}$, the optimal input state will be $|\psi\rangle=(|\psi_1\rangle+|\psi_2\rangle)/\sqrt{2}$, where $|\psi_i\rangle$ is the eigenvector corresponding to the eigenvalue $\lambda_i$ {of $G$. Then the final state} reads
\begin{align}
U|\psi\rangle&=\frac{1}{\sqrt{2}}\left( e^{-i\alpha\lambda_1}|\psi_1\rangle+e^{-i\alpha\lambda_2}
|\psi_2\rangle \right) \\
&=\frac{1}{2}\left(1+ e^{i\alpha(\lambda_1-\lambda_2)}\right) |\psi_+\rangle
+\frac{1}{2}\left(1-e^{i\alpha(\lambda_1-\lambda_2)}\right)
|\psi_-\rangle,
\end{align}
where $|\psi_\pm\rangle:=(|\psi_1\rangle\pm|\psi_2\rangle)/\sqrt{2}$.
The corresponding density operator has the following matrix representation 
(in the basis $|\psi_\pm\rangle$)
\begin{equation}
\left(\begin{array}{ccc}
\cos^2\left(\alpha\frac{(\lambda_1-\lambda_2)}{2}\right) & & i\sin\left(\alpha\frac{(\lambda_1-\lambda_2)}{2}\right) \cos\left(\alpha\frac{(\lambda_1-\lambda_2)}{2}\right)  \\ \\
 -i\sin\left(\alpha\frac{(\lambda_1-\lambda_2)}{2}\right) \cos\left(\alpha\frac{(\lambda_1-\lambda_2)}{2}\right) & & \sin^2\left(\alpha\frac{(\lambda_1-\lambda_2)}{2}\right)
\end{array}\right).
\end{equation}
It is easy to check, with methods of Sec.\ref{Sec:Fisher}, that the Fisher information achievable with this matrix,
can also be achieved with a matrix having the same diagonal terms and {off-diagonal terms} with different phases or even nullifying.
As a consequence, when estimating $\alpha$ in our bottleneck scheme, we would like to have 
a similar form for the reduced density matrix $\rho_B$.
This could come from a state 
\begin{equation}\label{formBF}
|\Psi_{BF}\rangle=\cos(.) |\Psi\rangle+e^{i\phi} \sin(.) |\Psi^\perp\rangle,
\end{equation}
where the argument of trigonometric functions must be proportional to $\alpha$ and $|\Psi\rangle$,
$|\Psi^\perp\rangle$ are orthogonal vectors in the space $\mathbb{C}^2\otimes \mathbb{C}^2$
(if $|\Psi\rangle$ is factorable the orthogonality condition must hold true at least in the subsystem $A$).

Suppose now to have found $|\Psi\rangle \in \mathbb{C}^2\otimes \mathbb{C}^2$ such that 
\begin{equation}\label{ocond1}
G|\Psi\rangle=a|\Psi^\perp\rangle, \quad a\in\mathbb{C},
\end{equation}
and hence 
\begin{equation}\label{ocond2}
G|\Psi^\perp\rangle=a^{*}|\Psi\rangle.
\end{equation}
Then, using the Taylor expansion, {it results} 
\begin{equation}
e^{-i\alpha G}|\Psi\rangle=\cos( |a|\alpha)\,|\Psi\rangle
-i e^{i \arg a}\sin( |a|\alpha) \, |\Psi^\perp\rangle,
\end{equation}
which is compatible with the form \eqref{formBF}.

So the problem can be reduced to find $|\Psi\rangle \in \mathbb{C}^2\otimes \mathbb{C}^2$
satisfying \eqref{ocond1}. To this end, we can look for eigenstates of operators anti-commuting with $G$.  In fact, if $GA+AG=0$ and $A|\Psi\rangle=\lambda |\Psi\rangle$ (with $\Re\{\lambda\}\neq 0$) it will be
\begin{equation}
\langle \Psi| \left(GA+AG\right) |\Psi\rangle=0
\Rightarrow
2\Re\{\lambda\} \langle \Psi | G |\Psi \rangle=0
\Rightarrow
{G|\Psi\rangle \perp |\Psi\rangle.}
\end{equation}
{Hence $ G|\Psi\rangle$ can be used in place of $|\Psi^\perp\rangle$,}
Summarizing, in order to find $\overline{J}_B$, we have to look for optimal input states among the eigenstates of operators anti-commuting with the generator.
Let us closely analyze a couple of cases.

\bigskip

Moving on from Sec.\ref{Sec:qubit1} the first case of generator where to apply this procedure seems 
\begin{equation}\label{firstGcase}
G=\sigma_1\otimes\sigma_1+ t_1 I\otimes \sigma_1+ t_2 \sigma_1\otimes I+ t_3 I\otimes I,
\qquad t_1,t_2,t_3\in\mathbb{R},
\end{equation}
with $t_3$ being a generic coefficient not {necessarily} equal to the product $t_1t_2$. 
However one can easily realize that $t_3$ does not appear in both $\overline{J}_{BF}$ and 
$\overline{J}_{B}$.
Hence the results will be the {same as} those found in Sec.\ref{Sec:qubit1}.

\bigskip

Next we are led to consider a generator of the kind
\begin{equation}\label{eq:Gcase}
G=\sigma_1 \otimes \sigma_1+t_1 I \otimes \sigma_3+t_2 \sigma_3 \otimes I, 
\end{equation}
which cannot be traced back to the tensor product of two generators.

Using the eigenvalues of \eqref{eq:Gcase} the maximum Fisher information one can get when measuring the system $BF$ results
\begin{equation}\label{eq:JBFt1t2}
\overline{J}_{BF}=4 \left(1+\left( |t_1|+|t_2|\right)^2 \right).
\end{equation}
For what concern the calculation of $\overline{J}_B$,
let us write the anticommutator $A$ as a generic Hermitian matrix 
\begin{equation}\label{matrixA}
A=\left(
\begin{array}{cccc}
a_{11} & a_{12}+i a_{21} & a_{13}+i a_{31} & a_{14}+i a_{41} \\
a_{12}-i a_{21} & a_{22} & a_{23}+i a_{32} & a_{24}+i a_{42} \\ 
a_{13}-i a_{31} & a_{23}-i a_{32} & a_{33} & a_{34}+i a_{43} \\ 
a_{14}-i a_{41} & a_{24}-i a_{42} & a_{34}-i a_{43} & a_{44}
\end{array}
\right).
\end{equation}
The solutions for $A$ anticommuting with \eqref{eq:Gcase} 
must be distinguished depending on the values of $t_1$ and $t_2$.

\begin{itemize}

\item[{\bf i)}]
$t_1\neq 0$ and $t_2=0$.

\begin{align}
&a_{14}=-a_{11}t_1,\\
&a_{23}=a_{22}t_1, \quad a_{24}=-a_{13}, \quad a_{42}=a_{31},\\
&a_{33}=-a_{22}, \quad a_{34}=-a_{12}-2 a_{13}t_1, \quad a_{43}=a_{21}+2 a_{31}t_1,\\
&a_{aa}=-a_{11}.
\end{align}
Then, upon normalization, the eigenvectors of $A$ can be cast into the following form
\begin{equation}\label{in-gen1}
\ket{\Psi_\pm(\theta,\phi)} = \left( \cos \theta \ket 0 \pm i \sin \theta \ket 1\right)_A \left( \frac{\ket 0+e^{i \phi} \ket1 }{\sqrt 2}\right)_E,
   \quad \theta,\phi\in[0,2\pi],
  \end{equation} 
which provides the expression for eigenvectors in Eq.\eqref{ocond1} with $a=\sqrt{1+t_1^2}$. 

This in turn gives $\overline{J}_B$ not depending on $t_1$ and equal to 4, 
thus by referring to \eqref{eq:JBFt1t2} we have
\begin{equation}
\Delta=\overline{J}_{BF}-\overline{J}_{B}=4(1+t_1^2)-4=4t_1^2,
\end{equation}  
that nullifies only for $t_1=0$.\footnote{It is worth mentioning that the set of optimal input states  when $t_1=t_2=0$ extends to
$\left( \cos \theta \ket 0 \pm i \sin \theta \ket 1\right)_A  \ket \varphi _E$,
$\forall \theta\in[0,2\pi]$ and $\forall \ket \varphi \in \mathbb{C}^2$.}


\item[{\bf ii)}]
$t_1=0$ and $t_2\neq 0$.

\begin{align}
&a_{14}=-a_{11}t_2,\\
&a_{23}=-a_{22}t_2, \quad a_{24}=-a_{13}-2a_{12}t_2, \quad a_{42}=a_{31}+2a_{21}t_2,\\
&a_{33}=-a_{22}, \quad a_{34}=-a_{12}, \quad a_{43}=a_{21},\\
&a_{aa}=-a_{11}.
\end{align}
Then, upon normalization, the eigenvectors of $A$ can be cast into the following form
\begin{equation}\label{eq:optimalii}
\ket{\Psi_{\pm\pm}(\theta)} = \left(  \frac{\ket 0 \pm i  \ket1 }{\sqrt 2}\right)_A\left( \cos \theta \ket 0 \pm \sin \theta \ket 1\right)_E  ,\quad \theta \in[0, 2\pi],
\end{equation}
which provides the expression for eigenvectors in Eq.\eqref{ocond1} with $a=\sqrt{1+t_2^2}$. 

This in turn gives $\overline{J}_B=4(1+t_2^2)$, which 
results equal to $\overline{J}_{BF}$ of
 \eqref{eq:JBFt1t2} implying $\Delta=0$.


\item[{\bf iii)}]
$t_1t_2\neq 0$.

\begin{align}
&a_{12}= a_{21}=a_{13}=a_{31}=0,  \quad
a_{14}=-a_{11}(t_1+t_2),\\
&a_{24}=a_{22}(t_1-t_2), \quad 
a_{42}=a_{32}, \\
&a_{33}=-a_{22}, \quad
a_{34}=a_{43}=0, \\
&a_{44}=-a_{11}.
\end{align}
Then, upon normalization, the eigenvectors of $A$ can be cast into the following form
\begin{align}
&\ket{\Psi_{\pm}(\theta)} = \cos\theta|01\rangle\pm i \sin\theta |10\rangle, 
\qquad  \theta\in[0,2\pi],\label{eq:evec1}\\
&\ket{\Psi_{\pm}(\theta)} = \cos\theta|00\rangle\pm i \sin\theta |11\rangle, 
\qquad  \theta\in[0,2\pi], \label{eq:evec2}  
\end{align}
which provide the expression for eigenvectors in Eq.\eqref{ocond1} 
with respectively $a=\sqrt{1+(t_1-t_2)^2}$ for \eqref{eq:evec1} and 
$a=\sqrt{1+(t_1+t_2)^2}$ for \eqref{eq:evec2}.

Taking either \eqref{eq:evec1}, or \eqref{eq:evec2} as an input we arrive, with
the technique of Sec.\ref{Sec:Fisher}, to the following Fisher information for the $B$ system
\begin{eqnarray}\label{JBt1t2}
{ J}_B=\frac{4 a^2\left[a \sin (2 \theta) \cos \left(2 \alpha  a\right)\pm\cos (2 \theta) \sin \left(2 \alpha  a\right)\right]^2}
{a^4-\left[a \sin (2 \theta) \sin \left(2 \alpha  a\right)
+\cos (2 \theta) \left(\cos \left(2 \alpha  a\right)+a^2-1\right)\right]^2}.
\end{eqnarray}
One can easily show that for each value of $\alpha$, there exists   
at least a value $\theta(\alpha)$ {giving ${ J}_B=4$.} 
Hence $\overline{J}_B=4$. This presumes, however, to adjust the input state according to the value of the parameter $\alpha$, which is in principle unknown. Thus we prefer to consider a unique input state for all $\alpha$.
In such a circumstance, from \eqref{JBt1t2} we argue that $\overline {J}_B=4$ 
only when $a=1$ and $\theta=0,\frac{\pi}{4},\frac{\pi}{2},\frac{3\pi}{4}\ldots$,
implying $ |t_1|=|t_2|$. 

In any case, from Eq.\eqref{eq:JBFt1t2} we have $\overline {J}_{BF}=4 $ only for $t_1=t_2=0$, hence we can conclude that $\Delta=0$ if $t_1=t_2=0$.
When $t_1,t_2\neq 0$ the quantity $\Delta$ depends on $\alpha$ through
 \eqref{JBt1t2}. Figure \ref{fig_contour} shows $\Delta$ 
vs $t_+:=t_1+t_2$ (assuming $t_1t_2 > 0$) and $\alpha$ for $\theta=\pi/4$. 
We can see that when $t_+$ approaches zero (i.e. $t_1,t_2\to 0$ given the assumption $t_1t_2 > 0$) the quantity $\Delta$ tends to zero.\footnote{Notice that in 
\eqref{JBt1t2} we cannot take the limit $t_1\to 0$ (or $t_2\to 0$), because otherwise we should consider the case {\bf ii)} (or {\bf i)} respectively). }
As soon as $t_+$ becomes different from zero, peaks appear whose width and number increases 
with $|t_+|$. 
\begin{figure}[H]
	\centering
	\includegraphics[width=0.5\textwidth]{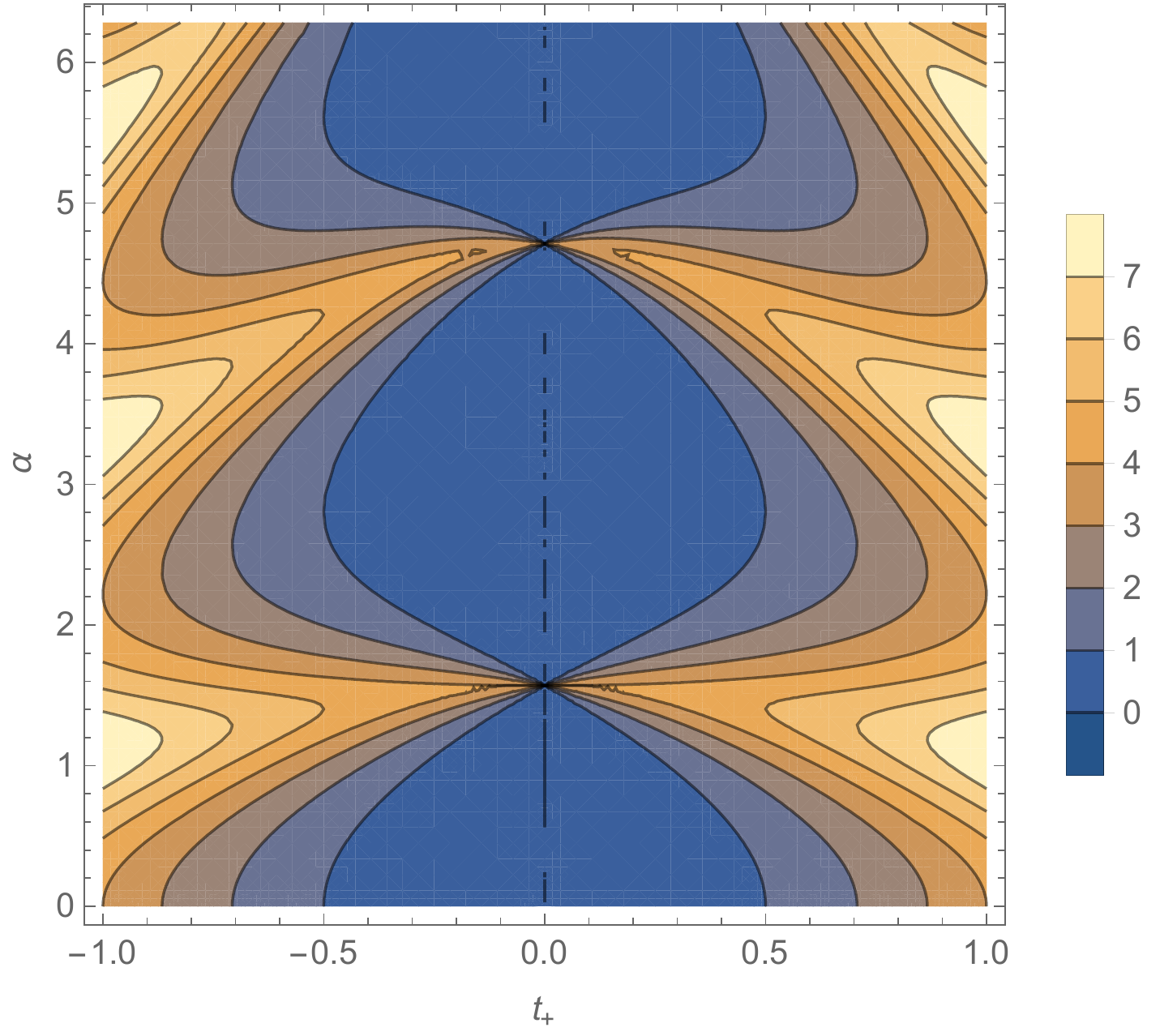}
	\caption{Contour plot of $\Delta$ vs $t_+:=t_1+t_2$ (assuming $t_1t_2 > 0$) and $\alpha$ for $\theta=\pi/4$.}
	\label{fig_contour}
\end{figure}


\end{itemize}


All together the results of {\bf i)},  {\bf ii)} and {\bf iii)} show that 
zero gap can only be attained when $t_1=0$, a tighter condition with  
respect to the case where $G=G_1\otimes G_2$ (see Theorem \ref{th:D0}).
Furthermore, according to \eqref{in-gen1}, \eqref{eq:optimalii}, \eqref{eq:evec1} 
and \eqref{eq:evec2}, factorable $AE$ states are enough to maximize $J_B$.
Although it is not guaranteed that the input states found following the method described at beginning of this Section are the only optimal ones, numerical search has shown that this is the case (see Appendix A).


\bigskip


{
Let us conclude with some considerations on more general forms of generator.
Actually, the most general form is $G=\sum_{i,j=0}^3c_{ij}\sigma_i\otimes\sigma_j$ 
($\sigma_0=I$ and $c_{ij}\in\mathbb{R}$),
which can also be recast into the form
\begin{equation}
c_{00} I\otimes I + c_{ \boldsymbol{\hat{m}}} \, 
\boldsymbol{\hat{m}}\cdot \boldsymbol{\sigma}\otimes I
+c_{ \boldsymbol{\hat{n}}} \, I \otimes \boldsymbol{\hat{n}}\cdot \boldsymbol{\sigma}
+ c_{ \boldsymbol{\hat{p}}} \, \boldsymbol{\hat{p}}\cdot \boldsymbol{\sigma}\otimes \sigma_1
+ c_{ \boldsymbol{\hat{q}}} \, \boldsymbol{\hat{q}}\cdot \boldsymbol{\sigma}\otimes \sigma_2
+ c_{ \boldsymbol{\hat{r}}} \, \boldsymbol{\hat{r}}\cdot \boldsymbol{\sigma}\otimes \sigma_3,
\end{equation}
with $\boldsymbol{\hat{m}}$, $\boldsymbol{\hat{n}}$, $\boldsymbol{\hat{p}}$, $\boldsymbol{\hat{q}}$, $\boldsymbol{\hat{r}}$ are generic directions in $\mathbb{R}^3$. 
Now, for what concern the attainability of $\overline{J}_{BF}$ by $\overline{J}_{B}$, the choice $c_{00}=0$ will not affect the results, because the identity is present in the second and the third term as well.
Furthermore, since ${ \boldsymbol{\hat{m}}} $ is a generic direction, we can take 
$ \boldsymbol{\hat{m}}\cdot \boldsymbol{\sigma}=\sigma_3$
and similarly $\boldsymbol{\hat{n}}\cdot \boldsymbol{\sigma}=\sigma_3$.
Finally, out of the three terms
$\boldsymbol{\hat{p}}\cdot \boldsymbol{\sigma}\otimes \sigma_1$,
$\boldsymbol{\hat{q}}\cdot \boldsymbol{\sigma}\otimes \sigma_2$,
$\boldsymbol{\hat{r}}\cdot \boldsymbol{\sigma}\otimes \sigma_3$,
it is enough to take only one,
because we have evidence that contributions along three different (although orthogonal) directions behave in the same way (see Appendix B).
Thus we can choose $\sigma_1\otimes\sigma_1$ and arrive to the example studied in Eq.\eqref{eq:Gcase} (choosing instead $\sigma_3\otimes\sigma_3$ leads to the example studied in Eq.  \eqref{firstGcase}).

This shows that the example studied in Eq.\eqref{eq:Gcase}, is representative of the most general generator, concerning the attainability of $\overline{J}_{BF}$ by $\overline{J}_{B}$. Hence, we can draw the following conjecture.
}

\begin{conjecture}\label{conj}
Given a family of two-qubit unitaries $U^{AE\to BF}_\alpha=e^{-i\alpha G}$,
with generator $G=\sum_{i,j=0}^3c_{ij}\sigma_i\otimes\sigma_j$ 
($\sigma_0=I$ and $c_{ij}\in\mathbb{R}$),
in order to have $\overline{J}_B=\overline{J}_{BF}$ it is sufficient that 
$c_{0j}=0$, $\forall j$, i.e.
${\rm Tr}_B G=0$.
\end{conjecture}


\subsection{Multiple instances estimation}

Similarly to Sec.\ref{Sec:Multiple} we shall consider here estimation by two copies of the unitary $U_\alpha$ arising from the generator \eqref{eq:Gcase}.
Given an input state $\ket\Upsilon_{in}$ for systems $A_1E_1A_2E_2$,
the {global final state} after unitary transformations reads
\begin{eqnarray}
\ket \Upsilon_{f}= \left( e^{-i\alpha G}\otimes e^{-i\alpha G}\right) \ket\Upsilon_{i}.
\end{eqnarray}
According to Sec.\ref{Sec:Fisher}, the maximum Fisher information we can get when accessing the {whole final system} reads 
\begin{equation}\label{maxJBF2}
\overline{J}_{\boldsymbol{BF}}= 16 \left(1+(|t_1|+|t_2|)^2\right).
\end{equation}
{However, the (output) state we are interested in is} 
\begin{eqnarray}
\rho_{\boldsymbol{B}}={\rm Tr}_{\boldsymbol{F}} |\Upsilon\rangle_{f}\langle \Upsilon |.
\end{eqnarray}
To compute the maximum Fisher information related to it we 
have {to refer} to the three cases analyzed in the previous Subsection.
 
\begin{itemize}

\item[{\bf i)}]
We expect the optimal input in two instances to be among the entangled states built up 
with twofold tensor product of states \eqref{in-gen1}. 
Numerical investigations (see Appendix A) show that {there is no one state} that gives 
$\overline{J}_{\boldsymbol{B}} = 4 \overline{J}_{B}$ for all $\alpha$ (unless $t_1=0$).
Therefore we focus on the possibility of having
$2 \overline{J}_{B} \leq {J}_{\boldsymbol{B}} \leq 4 \overline{J}_{B}$,
i.e. performance always better (or equal) than separable parallel instances.
This can be achieved with the following state
\begin{equation}
\ket \Upsilon_{i}=\frac{1}{\sqrt{2}}
\left(\ket{\Psi_+(0,0)}^{\otimes 2} + \ket{\Psi_+(0,\pi)}^{\otimes 2} \right),
\label{input22}
\end{equation}
which is maximally entangled between $E_1E_2$.


\item[{\bf ii)}]
Here we expect the optimal input in two instances to be among the entangled states built up 
with twofold tensor product of states \eqref{eq:optimalii}. 
Numerical investigations (see Appendix A) show that {there is no one state} that gives 
$\overline{J}_{\boldsymbol{B}} = 4 \overline{J}_{B}$ for all $\alpha$ (unless $t_2=0$).
{This, in turn, prevents} us from having 
$\Delta=0$ when going from
single to double instance. Although that might be surprising, it can be explained by considering the new generator 
$\Gamma$ resulting in double instance:
\begin{equation}
e^{-i\alpha G}\otimes e^{-i\alpha G}=e^{-i\alpha \Gamma},
\quad
\Gamma=I\otimes G+G\otimes I.
\end{equation}
As we can see it contains identity on the accessed subsystems, and hence by referring to Conjecture \ref{conj}, 
{the possibility of having $\Delta=0$
it no longer guaranteed.}

Thus we focus on the possibility of having also here
$2 \overline{J}_{B} \leq {J}_{\boldsymbol{B}} \leq 4 \overline{J}_{B}$,
i.e. performance always better (or equal) than separable parallel instances.
This can be achieved with the following state
\begin{equation}
\ket \Upsilon_{i}=\frac{1}{\sqrt{2}}
\left(\ket{\Psi_{++}(0)}^{\otimes 2} + \ket{\Psi_{--}(0)}^{\otimes 2} \right),
\label{input23}
\end{equation}
which is maximally entangled between $A_1A_2$.


\item[{\bf iii)}]
Also in this case we expect the optimal input in two instances to be among the entangled states built up 
with twofold tensor product of states  \eqref{eq:evec1} (or \eqref{eq:evec2}). Numerical search (see Appendix A) shows that 
for each value of $\theta$ the quantity \eqref{JBt1t2} can be quadruplicated with one such a state.
For example the state 
\begin{align}
\ket \Upsilon_{i}&=\frac{1}{\sqrt{2}}
\left( \ket{\Psi_+(0)}^{\otimes 2} + \ket{\Psi_+(\pi/2)}^{\otimes 2} 
\right),
\label{input22}
\end{align}
which is maximally entangled among all parties $A_1E_1A_2E_2$, 
gives 
\begin{equation}
{J}_{\boldsymbol{B}}=16 \frac{a^2 \cos^2(2a\alpha)}{a^2-1+\cos^2(2a\alpha)}.
\end{equation}
This is 4 times the quantity in Eq.\eqref{JBt1t2} with $\theta=\pi/4$.
Thus also the gap $\Delta$ is simply quadruplicated.

\end{itemize}

Summarizing, even if the generator $G$ satisfies the conditions to get $\Delta=0$, 
it is not guaranteed that this result can be attained over multiple instances too
(this is in contrast with tensor product generator where $\Delta=0$ can be kept over multiple instances by simply using entanglement across $A$ systems).
In order to minimize the gap various kind of entanglement in the input (across $A$ systems,
or across $B$ systems, {or fully}) might be necessary.


\section{Conclusion}\label{Sec:conclu}

In conclusion, given a one-parameter family of unitaries $\{U_\alpha^{AE\to BF}\}$,  
we considered {the estimation of the parameter $\alpha$} by accessing only the system $B$. 
The estimation capabilities have been related to the properties of unitaries' generators. 
First, the continuity of quantum Fisher information has been established with respect to them.
Then, conditions on the generators of two-qubit unitaries have been found to achieve the same quantum Fisher information {we would have} in the absence of bottleneck. 
These can be summarized as the generator $G$ 
satisfying ${\rm Tr}_BG = 0$, or in other words, only containing elements of the algebra 
$\mathfrak{su}(2)$ for the first qubit.
Whenever it can be written as tensor product $G_1\otimes G_2$, it is sufficient that  
$G_1$ belongs to the special unitary algebra. In this latter case also the necessary condition has been found.
When a gap appears, it depends on the strength of terms deviating from elements 
of the algebra $\mathfrak{su}(2)$ for the first qubit.
From the analyzed cases entangled inputs across the $AE$ systems seem not {always necessary to reach the goal (it is whenever $c_{i1},c_{1i},c_{j3},c_{3j}\neq 0$ for some $i,j>0$). } In contrast, entangled inputs across multiple estimation instances enhance the performance, although not always by the celebrated scaling of the number of instances squared.
In particular, this happens when $G=G_1\otimes G_2$, thus guaranteeing, in this case, the extendibility of zero gap over multiple instances. 

{The idea put forward of relating the continuity of quantum Fisher information 
to generators could be extended to one-parameter dynamical semigroups and their generators as well (see Appendix C).
This, in turn, could enable studies on when entangled probe states enhance estimation accuracy to sub-shot noise (or Heisenberg regime) in noisy dynamics.

On another side, since} the bottleneck model employed here {can} be regarded as a two senders and one receiver quantum channel, we expect this work will {be seminal for} studies of quantum multiple-access channel estimation \cite{QMAC}.
What remains valuable for further investigation in future work is to extend the analysis to higher and/or different subsystems dimensions and see how the gap varies in terms of such dimensions. 
Even the consideration of $U_\alpha:\mathscr{H}\to\mathscr{H}$ with 
$\mathscr{H}$ of prime dimension $D$, while accessing {a final system} of dimension $d<D$,
could open new interesting perspectives.


\section*{Acknowledgments}

The work of M.R. is supported by China Scholarship Council.


\section*{Appendix A}\label{Sec:ApA}

Following up Hurwitz parametrization \cite{Hur}, we can write $N$-qubit states as
\begin{equation}\label{eq:Hur}
\sum_{n=0}^{2^N-1} \nu_n\, |\,[n]_2\,\rangle,
\end{equation}
where $[n]_2$ stands for the binary representation of $n$.
{We also have}
\begin{eqnarray}
\nu_0&=&\cos\vartheta_{2^N-1},\\
{\nu_{n>0}}&=&e^{i\varphi_n}\cos_{2^N-1-n}\prod_{\ell=2^N-n}^{2^N-1} \sin\vartheta_\ell,
\end{eqnarray}
with
\begin{equation}
\vartheta_n\in[0,\pi/2], 
\quad 
\varphi_n\in[0,2\pi].
\end{equation}
Now searching the maximum of a function over the set of states \eqref{eq:Hur} can be done by randomly sampling such states according to the Haar measure of $U(2^N)$ \cite{ZS}.
{However, in such a way, we cannot account for separable states, as this subset of states has a vanishing probability measure \cite{DLMS}.}
Therefore we opted for sampling on a grid of 50 points for $\vartheta_n$ in $[0,\pi/2]$
and 200 points for $\varphi_n$ in $[0,2\pi]$.


\section*{Appendix B}\label{Sec:ApB}

{Consider a unitary with generator
\begin{equation}
G=\sigma_1\otimes\sigma_1+t_{22}\sigma_2\otimes\sigma_2+t_{33}\sigma_3\otimes\sigma_3,
\end{equation}
where $t_{22}, t_{33}\in \mathbb{R}$. The eigenvalues of $G$ result  
$\{-1-t_{22}-t_{33}, 1+t_{22}-t_{33}, 1-t_{22}+t_{33}, -1+t_{22}+t_{33}\}$, hence the maximum Fisher information we can get when accessing both systems $B$ and $F$ is:

\begin{equation}\label{eq:maxJBF}
\overline{J}_{BF}=\left\{ \begin{array}{ccc}
4\left(1+|t_{33}|\right)^2 & & |t_{22}|\leq|t_{33}|, \; |t_{22}|<1 \\
4\left(1+|t_{22}|\right)^2 & & |t_{33}|\leq|t_{22}|, \; |t_{33}|<1 \\
4\left(t_{22}+t_{33}\right)^2 & & |t_{22}|, |t_{33}|\geq 1, \; t_{22}t_{33}>0\\
4\left(t_{22}-t_{33}\right)^2 & & |t_{22}|, |t_{33}|\geq 1, \; t_{22}t_{33}<0\\
\end{array}\right..
\end{equation}
We can obtain:
\begin{itemize}
\item
$\overline{J}_{B}=\overline{J}_{BF}=4\left(1+t_{33}\right)^2$, with input
$\frac{1}{2}\left(\ket{00}+\ket{01}-\ket{10}+\ket{11} \right)$;

\item $\overline{J}_{B}=\overline{J}_{BF}=4\left(1-t_{33}\right)^2$, with input
$\frac{1}{2}\left(\ket{00}+\ket{01}+\ket{10}-\ket{11} \right)$;

\item $\overline{J}_{B}=\overline{J}_{BF}=4\left(t_{22}+t_{33}\right)^2$, with input 
$\frac{1}{2}\left(\ket{00}+\ket{01}-\ket{10}-\ket{11} \right)$;

\item $\overline{J}_{B}=\overline{J}_{BF}=4\left(t_{22}-t_{33}\right)^2$, with input
$\frac{1}{2}\left(\ket{00}+\ket{01}+\ket{10}+\ket{11} \right)$;

\item $\overline{J}_{B}=\overline{J}_{BF}=4\left(1+t_{22}\right)^2$, with input
$\ket{01}$;

\item $\overline{J}_{B}=\overline{J}_{BF}=4\left(1-t_{22}\right)^2$, with input 
$\ket{00}$.
\end{itemize}
}


\section*{Appendix C}\label{Sec:ApC}

{
\begin{corollary}
Given $\rho_i(\alpha)=e^{\alpha {\cal L}_i}\rho$,
with ${\cal L}_i$ Liuovillian superoperators, 
we have
{(dropping the dependence from 
$\alpha$ for a lighter notation):}
\begin{align}
\left| J(\rho_1)-J(\rho_2)\right|\leq 
\left(2\pi C_1+C_2+2\pi C_2 \min\left\{ \| {\cal L}_1\|_{1\to 1}, \| {\cal L}_2\|_{1\to 1} \right\}\right)
\left\| {\cal L}_1-{\cal L}_2 \right\|_{1\to 1},
\end{align}
where $C_1,C_2$ are as in Corollary \ref{cor:continuityG}, and $\|\cdot\|_{1\to 1}$ is the induced 
1-norm on the superoperators, i.e. $\left\| {\cal L}\right\|_{1\to 1}:=\sup_{\rho: \|\rho\|_1=1} \|{\cal L}\rho\|_1$.
\end{corollary}

\begin{proof}
Moving on from Theorem \ref{th:Jcontinuity}, for the first term at r.h.s. of Eq.\eqref{eq:Jcontinuity}, we have
\begin{align}
\left\|\rho_1-\rho_2\right\|{_2} &\leq  
\left\|\rho_1-\rho_2\right\|_1 \\
& \leq \left\| e^{\alpha {\cal L}_1}- e^{\alpha {\cal L}_2}\right\|_{1\to 1}
\label{eq:newcor2}\\
& \leq \alpha \left\| {\cal L}_1- {\cal L}_2\right\|_{1\to 1},
\label{eq:newcor3}
\end{align}
where from \eqref{eq:newcor2} to \eqref{eq:newcor3} we have used the property \eqref{eq:propertyexp} together with the fact that $e^{\alpha {\cal L}_i}$ is trace preserving. 

For the second term at r.h.s. of Eq.\eqref{eq:Jcontinuity}, instead, we have
\begin{align}
\left\|\partial_\alpha\rho_1- \partial_\alpha\rho_2\right\|_2
&=\left\| {\cal L}_1 e^{\alpha {\cal L}_1} - {\cal L}_2 e^{\alpha {\cal L}_2} \right\|_2 \\
&\leq \left\| {\cal L}_1 e^{\alpha {\cal L}_1} - {\cal L}_2 e^{\alpha {\cal L}_2} \right\|_1 \\
&\leq \left\| {\cal L}_1 e^{\alpha {\cal L}_1} - {\cal L}_2 e^{\alpha {\cal L}_2} \right\|_{1\to 1} \\
&\leq \left\| {\cal L}_1 e^{\alpha {\cal L}_1} 
-{\cal L}_1 e^{\alpha {\cal L}_2}
+{\cal L}_1 e^{\alpha {\cal L}_2}
- {\cal L}_2 e^{\alpha {\cal L}_2} \right\|_{1\to 1} \\
&\leq \left\| {\cal L}_1\right\|_{\to 1} \;
\left\|e^{\alpha {\cal L}_1} 
-e^{\alpha {\cal L}_2}\right\|_{1\to 1} 
+\left\|{\cal L}_1- {\cal L}_2\right\|_{1\to 1}\;
\left\| e^{\alpha {\cal L}_2} \right\|_{1\to 1} \label{eq:newcor4}\\
&\leq 
\left(1+\alpha \left\| {\cal L}_1\right\|_{1\to 1}\right)
\left\|{\cal L}_1- {\cal L}_2\right\|_{1\to 1},\label{eq:newcor5}
\end{align}
where from \eqref{eq:newcor4} to \eqref{eq:newcor5} we have used the property \eqref{eq:propertyexp} together with the fact that $e^{\alpha {\cal L}_i}$ is trace preserving. 
Notice that we could have reversed the role of $ {\cal L}_1$ and $ {\cal L}_2$.
Thus, by inserting Eqs.\eqref{eq:newcor3}, \eqref{eq:newcor5} into \eqref{eq:Jcontinuity} and taking into account that $\alpha\in[0,2\pi]$ we get the desired result.

\end{proof}
}


\end{document}